\documentclass[11pt,a4paper]{article}

\usepackage[utf8]{inputenc} 
\usepackage[T1]{fontenc}    %
\usepackage[english]{babel} 
\usepackage[final]{microtype} 

\usepackage{amssymb,amsmath,mathtools} 
\usepackage{amsthm} 

\usepackage{xcolor}
\usepackage{bbm}
\usepackage{braket}

\usepackage[hmargin=0.12\paperwidth,vmargin=0.19\paperwidth,bindingoffset=0cm]%
{geometry} 

\numberwithin{equation}{section} 

\theoremstyle{plain}
  
  \newtheorem{prop}{Proposition}

\theoremstyle{definition}
  
  \newtheorem{remark}{Remark}
  
  \numberwithin{prop}{section}
   \numberwithin{cor}{section}
   \numberwithin{remark}{section}

\title{\Large\bfseries Ensemble inter-relations  in random matrix theory}%
   
\author{Peter J. Forrester}
\date{}

\begin{document}

\maketitle

School of Mathematics and Statistics,  The University of Melbourne,
Victoria 3010, Australia. \: \: Email: {\tt pjforr@unimelb.edu.au}; \\

\bigskip

\begin{abstract}
\noindent 
The ensemble inter-relations to be considered are special features of classical cases, where the joint eigenvalue probability density can be computed explicitly. Attention will be focussed too on the consequences of these inter-relations, most often in relation to gap probabilities. A highlight from this viewpoint is an evaluation formula for the gap probability generating function for the circular orthogonal and circular symplectic ensembles (examples of Pfaffian point processes), in terms of the  gap probability generating function for real orthogonal matrices chosen with Haar measure ensembles (examples of determinantal point processes). The classes of inter-relations which lead to this result involve the superposition of the eigenvalue point processes of two independent classical ensembles, or consideration of the singular values of the classical ensembles with an evenness symmetry,
and the decimation operation of integrating over every second eigenvalue. It turns out that the probability density functions encountered through the consideration of superposition and decimation can in some cases be obtained via the consideration of various rank one matrix perturbations. These allow for different understandings of some of the inter-relations in this class. We make note too  of inter-relations --- which we consider generally as an identity or evaluation formula of a statistic linking two distinct random matrix ensemble --- that fall outside of this class, applying to the spectral form factor, and ones relating to discrete determinantal point processes. Not part of our review are the so-called duality relations of random matrix theory. Their consideration warrants a separate discussion.
\end{abstract}

\vspace{3em}

\section{Introduction}

For purposes of the present review, the term inter-relations in random matrix theory refers to an
identity, or evaluation formula for a statistic, which links two distinct random matrix ensembles. We say this
at the outset, as another class of inter-relations is the interplay between random matrix theory
and (seemingly) distinct fields of study, for example the statistical mechanics of log-potential Coulomb systems
\cite{Dy62}, \cite{Fo10}, \cite[\S 4]{BF22}, \cite{Se24}. Dualities are a special class of random matrix inter-relations, involving
the interchange of certain parameters. As perhaps the simplest example of the latter, consider the class
of Hermitian Wigner matrices $H = {1 \over 2} (X + X^\dagger)$, where the $N \times N$ random matrix $X$ has
all elements independently distributed with zero mean and variance $\sigma^2$ (for complex elements with zero mean,
$\sigma^2 := \langle |x_{ij} |^2 \rangle$). Denote the ensemble of such matrices by $\mathcal H_N$. Then one has
\cite[equivalent to Prop.~11]{FG06}
\begin{equation}\label{1.0}
\langle \det ( x \mathbb I_N - H) \rangle_{H \in \mathcal H_N} = i^{-N}
\langle  ( i x  - h)^N \rangle_{h \in \mathcal N(0,\sigma^2/2)},
 \end{equation}
 where $ \mathcal N(0,\sigma^2/2)$ denotes the zero mean normal distribution with variance $\sigma^2/2$. 
 In the context of a duality relation, the essential point here is the interchange of the pair $(N,n)$, where
on the left hand side $N$ is the matrix size, and $n=1$ is the power of the characteristic polynomial, while
on the right hand side the meaning of these parameters is reversed. For reasons of space,
duality type inter-relations will not be  considered further in this work.

As an introductory example of an inter-relation that is not also a duality formula, one can consult the
pioneering papers on random matrix theory by Dyson from the early 1960's, in particular the work
\cite{Dy62a}. In a previous paper \cite{Dy62} in the series, Dyson had introduced three distinguished
ensembles of random unitary matrices, referred to collectively as the circular ensembles. Of these two
are relevant for present purposes. One is the so called
circular unitary ensemble (CUE) consisting of matrices $U$ from the group $U(N)$ chosen with Haar
measure, while the other is the circular orthogonal
ensemble (COE) of symmetric unitary matrices; its members can be constructed in terms of CUE matrices by forming $U^T U$. A result from \cite{Dy62}
gives that the joint eigenvalue probability density function (PDF) of the circular ensembles is given by
\begin{equation}\label{1.1}
{1 \over Z_{N,\beta}} \prod_{1 \le j < k \le N} |
  e^{i \theta_k} - e^{i \theta_j} |^\beta,
   \end{equation}
with normalisation
 \begin{equation}\label{2.0a}
 Z_{N,\beta} = (2 \pi)^N {\Gamma(1 + N \beta/2) \over (\Gamma(1 + \beta/2))^N}.
  \end{equation}
  The cases $\beta = 1$ and 2 ($\beta$ is referred to as the Dyson index)
  correspond to the COE and CUE respectively.\footnote{The result
  for the CUE was in fact derived earlier in \cite{We39}; see \cite{DF17} for more historical details.} 
  
  Let $E_{N,\beta}(k;(0,\phi);{\rm CE}_N(\beta))$ denote the
  probability that there are exactly $k$ eigenvalues in the interval $(0,\phi)$ --- referred to as
  a (conditioned) gap probability --- in the circular ensemble
  with eigenvalue PDF (\ref{1.1}). Dyson was able to deduce that \cite[Eq.~(113)]{Dy62a}
   \begin{equation}\label{2.0b}
   E_{N,2}(0;(0,\phi);{\rm CUE}_N) =  E_{N,1}(0;(0,\phi))\Big (  E_{N,1}(0;(0,\phi);{\rm COE}_N) +  E_{N,1}(1;(0,\phi));{\rm COE}_N \Big ),
  \end{equation}
  so providing an inter-relation between gap probabilities in the CUE, and in the COE. 
  The structure of (\ref{2.0b})    led Dyson to consider the ensemble alt$\, ({\rm COE}_N \cup {\rm COE}_N)$. Here 
  $ {\rm COE}_N \cup {\rm COE}_N$ denotes the superposition of the eigenvalues of two independently chosen
  COEs, while alt says that only every second eigenvalue in this point process of $2N$ eigenvalues on the unit
  circle is
  to be observed. It was realised that (\ref{2.0b}) would result if it were true that\footnote{The symbol $\mathop{=}\limits^{\rm d}$
  denotes equal in distribution.}
    \begin{equation}\label{2.0c}
    {\rm alt} \,   ({\rm COE}_N \cup {\rm COE}_N) \mathop{=}\limits^{\rm d} {\rm CUE}_N,
    \end{equation}
    a result which was subsequently proved by Gunson \cite{Gu62}. This in turn feeds back into implying an
    infinite family of gap probability identities, generalising (\ref{2.0b}), which are most succinctly stated
    in terms of the generating functions\footnote{For purposes of comparison with earlier literature, we note that
  for reason stated in \cite[text below (3.33)]{Fo06} (observe  too the structure of (\ref{2.31A}) below), in this latter reference, $\pm$ on the RHS of the second generating
  function is replaced by $\mp$.}  
  \begin{align}\label{2.0d}    
{\mathcal E}^{{\rm CUE}_N}((0,\phi);\xi) & = \sum_{k=0}^\infty (1 - \xi)^k E_{N,2}(k;(0,\phi);{\rm CUE}_N) \nonumber \\
{\mathcal E}_{N,1}^{\pm }((0,\phi);\xi) & = \sum_{k=0}^\infty (1 - \xi)^k \Big ( E_{N,1}(2k;(0,\phi);{\rm COE}_N) +E_{N,1}(2k\pm 1;(0,\phi);{\rm COE}_N)  \Big ).
 \end{align}
 Thus one can deduce from (\ref{2.0c}) that  \cite{Me92}, \cite[Eq.~(3.29)]{Fo06}, \cite[Prop.~8.4.4]{Fo10}
   \begin{equation}\label{2.0e}
 {\mathcal E}^{{\rm CUE}_N}((0,\phi);\xi)  =   {\mathcal E}_{N,1}^{- }((0,\phi);\xi)  {\mathcal E}_{N,1	}^{+ }((0,\phi);\xi) 
 \end{equation}
 
 Our interest in this review is in random matrix inter-relations that are distinct from duality relations, the majority of which
 relate to superpositions of ensembles and operations such as alt, collectively referred to as decimations.
  In this circumstance, identities between conditioned gap probabilities will be seen to follow
 as a consequence. Outside of this class, but also considered, are inter-relations which follow by
 consideration of the spectral form factor, and inter-relations which give rise to discrete determinantal point
 processes.

\section{Inter-relations involving superpositions}
\subsection{Preliminary theory and notations}
The eigenvalue PDF of an ensemble of real symmetric (complex Hermitian) random matrices possessing
an orthogonal (unitary) symmetry has the functional form of being proportional to
  \begin{equation}\label{3.0}
 \prod_{l=1}^N w_\beta(x_l) \prod_{1 \le j < k \le N} |x_k - x_j|^\beta,
 \end{equation} 
 for $\beta = 1$ and 2 respectively. These are to be denoted OE${}_N(w_1)$ and UE${}_N(w_2)$, with $w_\beta$ referred to as
 the weight. For example, with the joint distribution of the matrices being taken as proportional to $e^{-(\beta/2) {\rm Tr} \, X^2}$
 (this corresponds to an independent Gaussian distribution on the diagonal and upper diagonal matrix entries), one observes
 the invariance under $X \mapsto V^\dagger X V$ for $V \in O(N)$ ($X$ real symmetric) and for $V \in U(N)$ ($X$ complex Hermitian).
 This gives rise to (\ref{3.0}) with $w_1(x) = e^{-x^2/2}$ and $w_2(x) = e^{-x^2}$; see \cite[Prop.~1.3.4]{Fo10}. At the matrix level,
 the names Gaussian orthogonal ensemble (GOE) and Gaussian unitary ensemble (GUE) are given to these particular random
 matrices, thus highlighting the invariances. Another case of interest is Hermitian matrices where each entry is a $2 \times 2$ block of complex numbers representing
 a quaternion; see \cite[\S 1.3.2]{Fo10}. This structure is preserved by the symplectic symmetry $X \mapsto V^\dagger X V$ for $V \in Sp(2N)$.
 Assuming that the joint element distribution too has this symmetry, the eigenvalue PDF of the $N$ independent
 eigenvalues (here the eigenvalues
 are doubly degenerate; see e.g.~\cite[Exercises 1.3 q.1]{Fo10}) has the functional form (\ref{3.0}) with $\beta = 4$, which is then
 denoted SE${}_N(w_4)$. In the Gaussian case the specific naming is GSE (Gaussian symplectic ensemble).
 
 The functional form (\ref{3.0}) contains too  eigenvalue PDFs equivalent to those of the circular ensembles. This comes about by
transforming from the unit circle to the real line according to the stereographic projection
  \begin{equation}\label{3.1}
  e^{i \theta_l} = {1 + i x_l \over 1 - i x_l}, \quad x_l = \tan {\theta_l \over 2},
 \end{equation} 
 where $x_l \to \pm \infty$ maps to $\theta_l \to \pm \pi$.
 Up to proportionality, (\ref{3.0}) results with
   \begin{equation}\label{3.1a}
   w_\beta(x) = {1 \over (1 + x^2)^{\beta(N-1)/2 + 1}}.
   \end{equation}
   For an ensemble of $2N$ eigenvalues, order the eigenvalues $x_{1} > x_2 > \cdots > x_{2N}$, and denote the corresponding
   eigenvalue PDF by ME${}_{2N,\beta}(w)$. Denote by even(ME${}_{2N,\beta}(w)$) 
   (odd(ME${}_{2N,\beta}(w)$) the ensemble obtained
   from the even (odd) labelled eigenvalues only. Denote by alt(ME${}_{2N,\beta}(w))$ the ensemble which is
  even(ME${}_{2N,\beta}(w))$ with probability ${1 \over 2}$, and    odd(ME${}_{2N,\beta}(w))$ with probability ${1 \over 2}$. 
  In terms of this terminology, the just
  specified stereographic projection of the circular ensembles, together with (\ref{2.0c}), implies
   \begin{equation}\label{3.0c}
    {\rm alt} \,   ({\rm OE}_N(w_1) \cup {\rm OE}_N(w_1)) \mathop{=}\limits^{\rm d} {\rm UE}(w_2), \quad w_1(x) = (1 + x^2)^{-(N+1)/2}, \: w_2(x) = (1 + x^2)^{-N}.
    \end{equation}
    In fact $w_1$ and $w_2$ herein are, up to a linear shift of variables, the unique weights on the real line with this property   \cite[Th.~4.6]{FR01}.
    
    Special attention is warranted in the case of (\ref{3.0}) with $\beta = 2$ and
  \begin{equation}\label{3.0d}   
  w_2(x) = x^a (1 - x)^b \mathbbm 1_{0 < x < 1},
  \end{equation}
  for particular $a,b$.
  Thus after the change of variables $x_l = \cos \theta_l$ $(0 < \theta_l < \pi)$, one obtains the PDF for the 
  independent eigenvalues
  $\{ e^{i \theta_l} \}_{1 \le j < N}$  $(0 < \theta_l < \pi)$ of  the particular real orthogonal groups with Haar measure
  $O^+(2N)$ ($a=b=-1/2$),  $O^+(2N+1)$ ($a=1/2$, $b=-1/2$), $O^-(2N+1)$ $(a=-1/2$, $b=1/2$) and 
  $O^-(2N+2)$ ($a=b=1/2$); see \cite[Prop 3.7.1]{Fo10}. 
    
    We introduce too the PDF on positive eigenvalues only, proportional to
   \begin{equation}\label{4.0}
 \prod_{l=1}^N w_2(x_l) \prod_{1 \le j < k \le N} (x_k^2 - x_j^2)^2,
 \end{equation}  
 which is to be denoted chUE$(w_2)$. This results from matrices with a unitary 
 symmetry for which the eigenvalues come in $\pm$ pairs.
  For example, with $Z$ an $n \times N$ ($n \ge N$) matrix of independent Gaussian entries ---
 a so-called rectangular complex Ginibre matrix \cite{BF22} --- the matrix
    \begin{equation}\label{4.0a} 
W = \begin{bmatrix} 0_{n \times n} & Z \\
Z^\dagger & 0_{N \times N}   \end{bmatrix}
 \end{equation}  
 is of this sort (it further has $n - N$ zero eigenvalues), with the positive eigenvalues having
 PDF    chUE$(w_2)$, $w_2(x) = x^{2(n - N) + 1}e^{-x^2}$; see
 e.g.~\cite[Prop.~3.1.3]{Fo10}. Other realisations of interest are $w_2(x) = e^{-x^2}$ 
 ($w_2(x) = x^2 e^{-x^2}$) which correspond to the PDF for the positive
 eigenvalues of $2N \times 2N$ ($(2N+1) \times (2N+1)$)
 anti-symmetric Gaussian random matrix with pure imaginary entries; see \cite[Table 3.1]{Fo10}. 
 One notes too that
 the eigenvalue PDF for the orthogonal ensembles
  can also be obtained by a stereographic change of variables (\ref{3.1}) in (\ref{4.0}) for particular
  $w_2$. For example
    \begin{equation}\label{4.0x}  
    {\rm chUE}_N \Big ( {x^2 \over (1 + x^2)^{2N+1}} \Big ) \equiv O^-(2N+2), \quad
     {\rm chUE}_N \Big ( {1 \over (1 + x^2)^{2N}} \Big ) \equiv O^-(2N+1).
     \end{equation} 
 
 \subsection{Superpositions of eigenvalues and  singular values}
 For a general Hermitian matrix ensemble ME${}_N(w)$ with both positive and negative eigenvalues, we denote
 by |ME${}_N(w_2)|$ the PDF for the absolute values of the eigenvalues, or equivalently its singular values.
 
 \begin{prop} \label{P2.1} (\cite[Eq.~(2.6)]{Fo06})
 Let $w_2(x)$ be even in $x$, and set $m = \lfloor N/2 \rfloor, \hat m = \lceil N/2 \rceil$. We have
  \begin{equation}\label{5.0}
  | {\rm UE}_N(w_2) | \mathop{=}\limits^{\rm d} {\rm chUE}_{\hat{m}}(w_2) \cup {\rm chUE}_{{m}}(x^2w_2).
  \end{equation}   
 \end{prop}
 
 \begin{proof}
 (Sketch) With $a(x)$ even, we first consider
  \begin{equation}\label{5.1}
  \Big \langle \prod_{l=1}^N a(x_l) \Big \rangle_{|{\rm UE}_N(w_2) |} =
 \Big \langle \prod_{l=1}^N a(x_l) \Big \rangle_{{\rm UE}_N(w_2) } \propto \det \bigg [
 \int_{-\infty}^\infty a(x) w_2(x) x^{j+k} \, dx \bigg ]_{j,k=0}^{N-1},
   \end{equation} 
 where the RHS follows from Andr\'eief's identity \cite{Fo19}. The assumption on both 
 $a(x), w(x)$ being even allows the RHS to be factored into two determinants. Using
 Andr\'eief's identity in the reverse direction on each of the latter identifies the RHS as
  \begin{equation}\label{5.2} 
   \Big \langle \prod_{l=1}^m a(x_l) \Big \rangle_{|{\rm chUE}_m(w_2) |} 
     \Big \langle \prod_{l=1}^{\hat{m}} a(x_l) \Big \rangle_{|{\rm chUE}_{\hat{m}}(w_2) |}. 
   \end{equation} 
   Since for $x>0$, $a(x)$ is arbitrary, the identity implied by (\ref{5.1}) and (\ref{5.2})   is
   equivalent to (\ref{5.0}).
  \end{proof}

  In the case $w_2(x) = e^{- x^2}$, by use of the theory recalled below (\ref{4.0a}), the identity
  (\ref{5.0}) relates the GUE to what may be abbreviated as the aGUE (anti-symmetric GUE); see also
  \cite{EL15}. An equivalent statement of (\ref{5.0}) in this case, obtained by a change of variables on the RHS,  is that
 \begin{equation}\label{5.0x}
 | {\rm GUE}_N | \mathop{=}\limits^{\rm d} \widetilde{{\rm LUE}}_{{\hat{m},-1/2}}   \cup   
 \widetilde{ {\rm LUE}}_{{{m},1/2}} ,
  \end{equation}  
 where LUE${}_{N,\alpha}$ --- in words the Laguerre unitary ensemble with $N$ eigenvalues and Laguerre
 parameter $\alpha$ --- refers to (\ref{3.0}) with $w_2(x) = x^\alpha e^{-x} \mathbbm 1_{x > 0}$, while
 $\widetilde{ {\rm LUE}}_{N,a}$ refers to the change of variables $x_l \mapsto x_l^2$ in the latter.
 Another viewpoint of the identity (\ref{5.0x}) is given in \cite{Du18}.
 
 With $E_{N,2}(k;J;w_2)$ ($\tilde{E}_{N,2}(k;J;w_2)$) denoting the probability that there are 
 exactly $k$ eigenvalues (singular values) in the interval $J$ for the ensemble UE${}_N(w_2)$,
 and with $\mathcal E_{N,2}^{ {\rm UE}_N(w_2) }(J;\xi)$ ($\mathcal E_{N,2}^{ |{\rm UE}_N(w_2) |}(J;\xi)$)
 denoting the corresponding generating functions (recall the first line of (\ref{2.0d})), it is immediate
 from (\ref{5.0}) that \cite[Eq.~(2.39) with $a=0$]{Fo10}
   \begin{equation}\label{2.0E}
 {\mathcal E}_{N,2}^{|{\rm GUE}_N|}((0,s);\xi)  =   {\mathcal E}_{N,2}^{{\rm LUE}_{\hat{m},-1/2}}((0,s^2);\xi)  {\mathcal E}_{N,2}^{{\rm LUE}_{{m},1/2}}((0,s^2);\xi). 
 \end{equation} 
 A consequence of (\ref{5.0x}) of a different type identified in \cite[Th.~2]{EL15} is a formula
 in distribution for the absolute value of the determinant of a matrix $H$ from GUE${}_N$. Thus with $\chi_k$ the
 square root of the classical $\chi^2$ distribution with $k$ degrees of freedom, one has
 \begin{equation}\label{5.0y} 
 | \det  H |  \mathop{=}\limits^{\rm d}   \prod_{j=1}^N \chi_{2 \lfloor j/2 \rfloor + 1}.
 \end{equation}   

In distinction to (\ref{3.0c}),
the inter-relation (\ref{5.0})    does not involve an auxiliary operation such as alt. We present
now inter-relations involving superposition which do involve an auxiliary operation, 
specifically that of even$(\cdot)$ defined below (\ref{3.1a}) \cite{FR01}.

\begin{prop}\label{P2.2} 
For $a,b > -1$, ${\rm Re} (\alpha) > -1$ define the pairs of weights
 \begin{multline}\label{5.0z}  
 (w_1(x), w_2(x)) =  \\
 \begin{cases} (e^{-x^2/2}, e^{-x^2}) & {\rm Gaussian} \\
 (x^{(a-1)/2} e^{-x/2} \mathbbm 1_{x>0}, x^a e^{-x}  \mathbbm 1_{x>0}) & {\rm Laguerre} \\
 ((1+x)^{(a-1)/2} (1-x)^{(b-1)/2}  \mathbbm 1_{-1<x<1},
(1+x)^{a} (1-x)^{b}  \mathbbm 1_{-1<x<1} ) & {\rm Jacobi} \\
((1+ix)^{-(N+\alpha+1)/2} (1-ix)^{-(N+\bar{\alpha}+1)/2},
(1+ix)^{-(N+\alpha)} (1-ix)^{-(N+\bar{\alpha})})
  & {\rm Cauchy.} 
\end{cases}
\end{multline}
In each of these cases the inter-relation
 \begin{equation}\label{5.3} 
 {\rm even} \, ( {\rm OE}_N(w_1) \cup {\rm OE}_{N+1}(w_1) ) \mathop{=}\limits^{\rm d} {\rm UE}_N(w_2)
 \end{equation} 
 holds true. Moreover, up to linear fractional change of variables, these pairs of weights uniquely
 satisfy (\ref{5.3}).
\end{prop}

\begin{proof} (Comments only) 
Fundamental to the superimposed ensemble $ {\rm OE}_N(w_1) \cup {\rm OE}_{N+1}(w_1)$ is an
identity which exhibits that the dependence on the even and odd numbered ordered eigenvalues
factorises. Thus with $S = \{ s_1, s_2,\dots, s_l \}$, $s_1 > s_2 \cdots > s_l \ge 1$ a set of positive
integers and 
 \begin{equation}\label{5.3'} 
\Delta(x_S) := \prod_{1 \le j < k \le l} (x_{s_j} - x_{s_k} ),
 \end{equation} 
one has \cite{Gu62,FR01}
 \begin{equation}\label{5.3a} 
\sum_{S \subset \{1,\dots,2N+1\} \atop |S| = N}
\Delta(x_S) \Delta(x_{\{1,\dots,2N+1\}-S})  = 
2^N \Delta(x_{\{1,3,\dots,2N+1\}}) \Delta(x_{\{2,4,\dots,2N\}}).
 \end{equation} 
 The RHS now has to be weighted by
  $ \prod_{l=1}^{2N+1} w_1(x_l)$ and  integration performed over the odd labelled
  eigenvalues, taking into consideration that  the domain of integration  interlaces the odd with
  the even
  labelled eigenvalues (and in the case of $x_1$ and $x_{2N+1}$,  the endpoints of support of
  $w_1$). For the Jacobi weight this amounts to a special case of the Dixon-Anderson integral
  \cite[Eq.~(4.15)]{Fo10} (cf.~\cite[Eq.~(6)]{Di05})
\begin{multline}\label{DA}
\int_X d \lambda_1 \cdots
d\lambda_N \, \prod_{1 \le j < k \le N} ( \lambda_j
- \lambda_k ) \prod_{j=1}^N \prod_{p=1}^{N+1} | \lambda_j - a_p|^{s_p - 1}
\\
= {\prod_{i=1}^{N+1} \Gamma ( s_i) \over \Gamma(\sum_{i=1}^{N+1} s_i )}
\prod_{1 \le j < k \le N +1} | a_k - a_j|^{s_j + s_k - 1},
\end{multline}
where $X$ is the region
 \begin{equation}\label{DA1}  
a_1 > \lambda_1 >  a_2 > \cdots > \lambda_N > a_{N+1}.
   \end{equation} 
 \end{proof}
 
 \begin{remark}
 With the notation $\Delta(u_1,\dots,u_N) := \prod_{1 \le j < k \le N} (u_j - u_k)$ (cf.~(\ref{5.3'})),
  $m, \hat{m}$ as in Proposition \ref{P2.1}, and $x_j = \sigma_{2j-1}$, $y_j=\sigma_{2j}$ $(j=1,2,\dots)$,
  where $\sigma_1 > \sigma_2 > \cdots > \sigma_N \ge 0$,
 an analogue of (\ref{5.3a}) is the identity \cite[Eq.~(11)]{BL15}
  \begin{equation}\label{5.3a+} 
  \sum_{\epsilon_1,\dots,\epsilon_N = \pm 1} | \Delta(\epsilon_1 \sigma_1,\dots,\epsilon_N \sigma_N) |
  = 2^N \Delta(x_1^2,\dots,x_{\hat{m}}^2) y_1 \cdots y_m
   \Delta(y_1^2,\dots,y_{{m}}^2). 
  \end{equation}
  In the case that the weight $w_1(x)$ is even in $x$, one identifies the LHS of (\ref{5.3a+}),
  after multiplication by $\prod_{l=1}^N w(\lambda_l)$,
   with
  $| {\rm OE}_N(w_1)|$. Further, if such a $w_1$ should be chosen from (\ref{5.0z})
  (and is thus a Gaussian, Jacobi with $a=b$, or Cauchy with $\alpha = \bar{\alpha}$ weight),
  it has been shown in \cite[Th.~1]{BF16} that the $\{x_j\}$ can be integrated over to deduce
  \begin{equation}\label{5.3a+1} 
  {\rm even} \, | {\rm OE}_N(w_1)|  \mathop{=}\limits^{\rm d} {\rm chUE}_m(x^{2 \mu} w_2).
   \end{equation}
   (In the Jacobi case the required integration follows by a suitable change of variables and
   specialisation of the parameters in (\ref{DA}).)
   Here $\mu$ takes on the value $0,1$ according to the requirement that $N = 2m + \mu$.
   Moreover, combining this inter-relation with (\ref{5.0}) gives that for these same even
   weights \cite[Eq.~(6)]{BF16}
  \begin{equation}\label{5.3a+2} 
| {\rm UE}(w_2) |   \mathop{=}\limits^{\rm d} {\rm even} \, | {\rm OE}_N(w_1)| \cup
 {\rm even} \, | {\rm OE}_N(w_1)| 
  \end{equation}
(cf.~(\ref{5.3})).
 \end{remark}
 
 An immediate combinatorial consequence of (\ref{5.3}) for the conditioned gap probabilities (in an obvious
 notation) is the inter-relation \cite[Eq.~(17)]{BF16}
  \begin{equation}\label{5.3b} 
  E_{N,2}(k;J_s;w_2) =
  \sum_{j=0}^{2k+1}  E_{N,1}(2k+1-j;J_s;w_1) \Big ( 
  E_{N+1,1}(j;J_s;w_1) + 
  E_{N+1,1}(j-1;J_s;w_1) \Big ).
 \end{equation} 
 Here $J_s$ is a single interval involving one endpoint of the support of the weight.  It is noted in
 \cite[Remark above Th.~6]{BF16} that, due to the effect of the freezing of a single eigenvalue
 in the superimposed ensemble, taking $a \to - 1^+$ in the Laguerre or Jacobi cases  of 
 Proposition \ref{P2.2} gives that for the pair of weights
  \begin{equation}\label{5.3c}   
 (w_1(x), w_2(x)) =  
 \begin{cases}
 ( e^{-x/2} \mathbbm 1_{x>0},  e^{-x}  \mathbbm 1_{x>0}) & {\rm Laguerre} \\
 ( (1-x)^{(b-1)/2}  \mathbbm 1_{-1<x<1},
 (1-x)^{b}  \mathbbm 1_{-1<x<1} ) & {\rm Jacobi} 
\end{cases}
\end{equation} 
 one has the inter-relation
 \begin{equation}\label{5.3x} 
 {\rm even} \, ( {\rm OE}_N(w_1) \cup {\rm OE}_{N}(w_1) ) \mathop{=}\limits^{\rm d} {\rm UE}_N(w_2),
 \end{equation} 
 derived independently in \cite{FR01}. This reasoning also implies that taking the limit $\alpha \to -1^+$
 in the Cauchy case of (\ref{5.3}) reclaims (\ref{3.0c}). 
 
 Reversing the sign of $x$ in the Jacobi case of (\ref{5.3c}) gives
\begin{equation}\label{2.26a} 
 {\rm odd} \, ( {\rm OE}_N(w_1) \cup {\rm OE}_{N}(w_1) ) \mathop{=}\limits^{\rm d} {\rm UE}_N(w_2),
 \end{equation} 
 for 
 $$(w_1(x), w_2(x)) = ( (1+x)^{(b-1)/2}  \mathbbm 1_{-1<x<1},
 (1+x)^{b}  \mathbbm 1_{-1<x<1} ).$$ 
 Also, one has
 \begin{equation}\label{2.26b} 
 {\rm odd} \, ( {\rm OE}_{N-1}(w_1) \cup {\rm OE}_{N}(w_1) ) \mathop{=}\limits^{\rm d} {\rm UE}_N(w_2),
 \end{equation} 
 for the particular Jacobi weights 
 $$(w_1(x),w_2(x)) =  (   \mathbbm 1_{-1<x<1},
   \mathbbm 1_{-1<x<1} ).$$
   This corresponds to the Dixon-Anderson integral (\ref{DA}) with each $s_p=1$.
    As for (\ref{5.3}), the above given pairs of weights for
   which (\ref{5.3x}), (\ref{2.26a}) and (\ref{2.26b}) hold were shown in
   \cite{FR01} to be unique up to fractional linear transformation. We remark in
   the (abstracted) setting that the eigenvalues are restricted to the integer lattice,
   the identity (\ref{2.26b}) is known in the theory of random tilings \cite{FFN10},
    \cite[proof of Prop.~10.2.3]{Fo10}; see too \cite[\S 5]{MOW09}.

 \subsection{Alternative constructions of interlaced ensembles}\label{S2.3}
 The identity (\ref{5.3a}) reveals that superimposing two real symmetric ensembles
 with orthogonal symmetry gives rise, up to the interlacing constraint, to a separation of the
 interaction between the odd and  even labelled eigenvalues. Following \cite{FR02b},
 here we will outline some different, matrix theoretic, constructions of this effect and show
 too how they imply inter-relation identities upon a suitable  decimation operation.
 
 We will consider first the Gaussian case of (\ref{5.3}). Let $A$ be an $N \times N$ GUE
 matrix. Now double the size of $A$ by representing each (typically complex) scalar as a $2 \times 2$ real matrix
  \begin{equation}\label{Q} 
 x + i y \mapsto \begin{bmatrix} x & y \\ - y & x \end{bmatrix}.
  \end{equation} 
 With the resulting matrix denoted $\tilde{A}$, one notes that the eigenvalues of $\tilde{A}$ are
 the eigenvalues of $A$, each with multiplicity two. Next introduce the $(2N+1) \times (2N+1)$
 bordered matrix
 \begin{equation}\label{7.3y}  
 M = \begin{bmatrix} \tilde{A} & \mathbf x \\
 \mathbf x^\dagger & a \end{bmatrix},
 \end{equation}
 where $\mathbf x$ is an $N \times 1$ column vector with entries $\sqrt{b}$ times a standard
 real Gaussian, and $a$ is $\sqrt{2b}$ times a standard real Gaussian. 
 Due to the entries being Gaussians, the matrix $\tilde{A}$ can be brought to its diagonal form
 without altering the distribution of the bordered entries. The characteristic polynomial of this
 simplified matrix is easy to compute, telling us that the matrix $M$ has  the $N$ eigenvalues of
 $A$, $\{y_j\}_{j=1}^{N}$ say, and $N+1$ eigenvalues given by the zeros of the rational
 function
 \begin{equation}\label{R1}
 R_1(\lambda) := \lambda - a + b \sum_{j=1}^N {w_j \over y_j - \lambda},
   \end{equation}
 where each $w_j$ is distributed as an independent gamma random variable with
 scale $\theta = 1$ and rate $k=1$. Up to a scaling of variables, the PDF for the roots
 of $R_1(\lambda)$ is known from \cite[Cor.~4]{FR02b}, allowing for the explicit functional
 form of eigenvalue PDF of $M$ to be computed.
 
 \begin{prop}\label{P2.3} (\cite[\S 6.1]{FR02b}
 Let the eigenvalues of the matrix $M$ in (\ref{7.3y}) be denoted $\{y_j\}_{j=1}^N \cup
 \{x_j\}_{j=1}^{N+1}$ ordered so that
 \begin{equation}\label{7.3y1}   
 x_1 > y_1 > \cdots > y_N > x_{N+1}.
 \end{equation}
 We have that the eigenvalue PDF of $M$ is proportional to
 \begin{equation}\label{7.3y2} 
 \prod_{l=1}^{N+1} e^{- c_1 x_l^2/2} \prod_{1 \le j < k \le N + 1} (x_j - x_k)
   \prod_{l=1}^{N} e^{- c_2 y_l^2/2}  \prod_{1 \le j < k \le N } (y_j - y_k),
   \end{equation}
where    $c_1 = {1 \over 2b}$ and $c_2 = - {1 \over 2 b} + 2$. Furthermore
 \begin{equation}\label{7.3y3} 
 {\rm even} \, (M ) \mathop{=}\limits^{\rm d} {\rm GUE}_N.
  \end{equation}
  \end{prop}
 
 The decimation result follows immediately from the property of $M$ noted below (\ref{7.3y}), and is
 independent of the scales $c_1,c_2$ in (\ref{7.3y2}). Making the particular choice $b={1 \over 2}$ gives
 for the latter $c_1 = c_2 = 1$, and we see from  (\ref{5.3a}) that (\ref{7.3y2}) is identical to the
 eigenvalue PDF of GUE${}_N \cup {\rm GUE}_{N+1}$. Thus the above construction provides a different
 viewpoint on the validity of (\ref{5.3x}) in the Gaussian case.
 
 It's similarly possible to give a matrix theoretic derivation of  (\ref{5.3x}) in both the Laguerre and Jacobi cases,
 upon knowledge that the eigenvalue PDF of $ {\rm even} \, ( {\rm OE}_N(w_1) \cup {\rm OE}_{N}(w_1) ) $, 
 up to the interlacing constraint, separates  the
 interaction between the odd and  even labelled eigenvalues \cite{FR02b}. For definiteness, we consider
 the Laguerre case. For $X$ an $n \times N$ $(n \ge N)$ complex Gaussian matrix, the positive definite matrix
 $A= X^\dagger X$ is referred to as a complex Wishart matrix, and such matrices are a realisation of the Laguerre
 unitary ensemble LUE${}_{N,\alpha}$, with Laguerre parameter $\alpha = n - N$. Now append $b > 0$ times a row of complex
 Gaussians to $X$ to form the matrix $Y$, and consider the matrix $B = Y^\dagger Y = A + b \mathbf x^\dagger \mathbf x$,
 where $\mathbf x$ is   an $N \times 1$ column vector with entries independent standard
 complex Gaussians. From the second equality, and the fact the entries are all Gaussians, the matrix $A$ can be brought to
 diagonal form (eigenvalues $a_1 > a_2 > \cdots > a_N > 0$ say) without altering the distribution of $\mathbf x$. Consideration
 of the characteristic polynomial shows that the eigenvalues  of $B$ are given by the zeros of the random rational function
 \begin{equation}
 R_2(\lambda) = 1 + b \sum_{l=1}^N  {w_j \over a_j - \lambda },
   \end{equation}
   where each $w_j$ is distributed as in (\ref{R1}). The PDF for the roots of $R_2$ is
   given by a specialisation of the parameters in \cite[Cor.~3]{FR02b}, which provides for the explicit
   functional form of 
 the joint  eigenvalue
 PDF of $A$ and $B$.
 
  \begin{prop} (\cite[consequence of Th.~4]{FR02b}
  In the above setting, and with $\{b_j\}_{j=1}^N$ the eigenvalues of $B$ (suitably ordered),
  the joint eigenvalue PDF of $A$ and $B$ is proportional to
   \begin{equation}\label{7.4}
 \prod_{l=1}^N a_l^\alpha e^{- a_l} e^{-(1/b) (b_l - a_l)}
 \prod_{1 \le j < k \le N} (a_j - a_k) (b_j - b_k),
  \end{equation}
 subject to the interlacing
    \begin{equation}\label{7.4a}
 b_1 > a_1 > \cdots > b_N > a_N.
 \end{equation} 
 Furthermore, from the interlacing and knowledge of the joint eigenvalue PDF of $\{a_j\}$, 
   \begin{equation}\label{7.4b}
   {\rm even} \, (B)  \mathop{=}\limits^{\rm d} {\rm LUE}_{N,\alpha}.
  \end{equation}   
\end{prop} 

We remark that the decimation result (\ref{7.4b}), which one notes holds for all $b > 0$, was
first obtained indirectly in \cite{BR01a}.
In the case $\alpha = 0$, $b = 2$, up to a simple scaling,
one recognises (\ref{7.4}) as the functional form of the 
eigenvalue PDF of the superimposed ensemble
${\rm LUE}_{N,0} \cup {\rm LUE}_{N,0}$.  Hence (\ref{7.4b}) reclaims the Laguerre case
of  (\ref{5.3x}).

 \section{Inter-relations involving decimation}
 \subsection{Singular values and the generating function $\mathcal  E^{{\rm COE}_N}((0,\phi);\xi) $}
 We have already seen a number of inter-relations between random matrix
 ensembles which involve the decimation operations of alt, even or odd.
 However, the focus in the previous section was in the circumstance that
 these were combined with the superposition operation. The identity (\ref{5.3a+1})
 is an example of an inter-relation involving only a decimation. The corresponding
 gap probability identity is simpler than (\ref{5.3b}). It reads \cite[Th.~4]{BF16}
  \begin{equation}\label{5.3z} 
E_{N,1}(2k+\mu - 1; (-s,s);w_1) +   E_{N,1}(2k+\mu ; (-s,s);w_1) =
E_{N,2}(k; (0,s^2);x^{\mu - {1 \over 2}} w_2(x^{1/2}) \mathbbm 1_{x > 0}).
 \end{equation}  
There is particular interest in (\ref{5.3a+1}) for the Cauchy weight with $\alpha = 0$,
which upon the stereographic projection (\ref{3.1}), OE${}_N(w_1)$ becomes COE${}_N$.
The meaning of $|{\rm COE}_N|$ is the point process on the half circle $0 < \theta < \pi$
formed by the eigen-angles of COE${}_N$ in this range, union the negative of the eigen-angles
in the range  $-\pi < \theta < 0$. Thus, upon making use too of  (\ref{4.0x}), one has \cite[Eq.~(61),  after correction of a sign on RHS]{BF16}
\begin{equation}\label{2.27}
{\rm even} ({\rm COE}_{N}) = O^{-}(N+1)
  \end{equation}
and consequently
\cite[Eq.~(3.25)  for $N$ even case]{Fo06}, \cite[Th.~6, , after correction of a sign on RHS]{BF16}
\begin{equation}
E_{N,1}(2k+\mu-1;(-\theta,\theta);{\rm COE}_{N}) + E_{N,1}(2k+\mu;(-\theta,\theta);{\rm COE}_{N})  =
 E_{\lfloor N/2 \rfloor ,2}(k;(0,\theta); O^-(N+1) ). 
  \end{equation}
  In terms of generating functions, upon introducing
\begin{equation}\label{2.29}
\mathcal E^{O^-(N+1)}((0,\theta);\xi) := \sum_{k=0}^\infty   (1 - \xi)^k E_{\lfloor N/2 \rfloor ,2}(k;(0,\theta); O^-(N+1) ),
 \end{equation}
and recalling (\ref{2.0d}) gives that \cite[Eq.~(3.26) for first equation]{Fo06}, 
\begin{equation}\label{2.30}
\mathcal E^{-}_{2N,1}((-\theta,\theta);\xi)  = \mathcal E^{O^-(2N+1)}((0,\theta);\xi), \quad \mathcal E^{+}_{2N-1,1}((-\theta,\theta);\xi)  = \mathcal E^{O^-(2N)}((0,\theta);\xi).
 \end{equation}    
 
 The identity (\ref{5.3a+1}) in the limit $\alpha \to -1^+$ of the Cauchy case, due to the freezing
 of an eigenvalue effect (recall text above (\ref{5.3c})), transforms  (\ref{5.3a+1}) to read
 \begin{equation}\label{5.3a+1a} 
  {\rm odd} \, | {\rm OE}_N(w_1)|  \mathop{=}\limits^{\rm d} {\rm chUE}_{\hat{m}}(x^{2 \hat{\mu}} w_2),
   \end{equation}
   where $\hat{\mu}$ is such that $N+1=\hat{m} + \hat{\mu}$ and with $(w_1,w_2)$ as in (\ref{3.0c}).
   This gives for the companion of (\ref{2.27}) \cite[Eq.~(62), after correction of a sign on RHS]{BF16}
  \begin{equation}\label{2.27+}
{\rm odd} ({\rm COE}_{N}) = O^{+}(N+1)
  \end{equation}
and the corresponding gap probability inter-relation 
 \cite[Th.~6,  after correction of a sign on RHS]{BF16}
\begin{equation}
E_{N,1}(2k+1-\mu;(-\theta,\theta);{\rm COE}_{N}) + E_{N,1}(2k-\mu;(-\theta,\theta);{\rm COE}_{N})  =
 E_{\lceil N/2 \rceil ,2}(k;(0,\theta); O^+(N+1) ). 
  \end{equation} 
In a notation analogous to (\ref{2.29}) we thus have the companion to (\ref{2.30})
\cite[Eq.~(8.149) for first equation]{Fo10},
\begin{equation}\label{2.30a}
\mathcal E^{+}_{2N,1}((-\theta,\theta);\xi)  = \mathcal E^{O^+(2N+1)}((0,\theta);\xi), \quad 
 \mathcal E^{-}_{2N-1,1}((-\theta,\theta);\xi)  = \mathcal E^{O^+(2N)}((0,\theta);\xi).
 \end{equation} 
 
 Recalling the definition of $\mathcal E^{\pm}_{N,1}$ from (\ref{2.0d}), the generating function identities
 (\ref{2.30a}) and  (\ref{2.30}) can be used to evaluate the generating function for
 $\{ E_{N,1}(k;(0,\phi);{\rm COE}_N) \}$,
 \begin{equation}\label{2.30b}
\mathcal E^{{\rm COE}_N}((0,\phi);\xi) :=  \sum_{k=0}^\infty  (1 - \xi)^k E_{N,1}(k;(0,\phi);{\rm COE}_N) 
\end{equation}
in terms of the corresponding generating function for the orthogonal groups.
 
 \begin{prop} (\cite[Eq.~(8.150) $N$ even case]{Fo10}, \cite[Prop.~5.1 \& Eq.~(5.10)]{BFM17})
 Let $\hat{\xi} := 2 \xi - \xi^2$ so that $1 - \hat{\xi} = (1 - \xi)^2$. We have
  \begin{equation}\label{2.31}
\mathcal  E^{{\rm COE}_N}((0,\phi);\xi) = {(1 - \xi)  \mathcal  E^{O^\nu(N+1)}((0,\phi/2);\hat{\xi})   + \mathcal  E^{O^{-\nu}(N+1)}((0,\phi/2);\hat{\xi})  \over 2 - \xi},
  \end{equation} 
  where $\nu = (-)^N$.
  \end{prop}
  
  \begin{remark} $ $ \\
  1.~Using (\ref{2.30a}) and  (\ref{2.30}) in (\ref{2.0e}) shows
   \begin{equation}\label{2.31a}
   \mathcal  E^{{\rm CUE}_N}((0,\phi);\xi) =  \mathcal E^{O^+(N+1)}((0,\phi/2);{\xi})   \mathcal E^{O^-(N+1)}((0,\phi/2); {\xi}). 
   \end{equation} 
   More generally one has \cite{Ra03}
    \begin{equation}\label{2.31b}
    | {\rm CUE}_N| \mathop{=}^{\rm d} O^+(N+1) \cup O^-(N+1),
   \end{equation} 
   which can be deduced from the choice $w(x) = (1 + x^2)^{-N}$ in (\ref{5.0}) and applying a stereographic projection. \\
   2. Denote by $({\rm CUE}_N)^p$ the eigenvalue ensemble formed by raising each eigenvalue in CUE${}_N$ to the
   $p$-th power $(p=1,2,\dots)$. Then it was shown in \cite{Ra03} that
     \begin{equation}\label{2.31c}
   ({\rm CUE}_N)^p    \mathop{=}^{\rm d}   \bigcup_{j=0}^{p-1} {\rm CUE}_{\lceil {(N-j) \over p} \rceil}.
    \end{equation} 
    In particular, this result shows that for $p \ge N$ the $p$-th power of the eigenvalues are independently distributed.
    This latter result is known too for the Ginibre ensemble of non-Hermitian Gaussian random matrices
    \cite{HKPV08}, and moreover a superposition result analogous to (\ref{2.31c}) can be derived for this ensemble too \cite{Du18}. 
   \end{remark}
  
  \subsection{Applications of the evaluation formula for $\mathcal  E^{{\rm COE}_N}((0,\phi);\xi)$}\label{S3.2}
  The eigenvalues of the COE form a Pfaffian point process, while those of the orthogonal group are
  examples of the simpler determinantal point processes; see \cite[Ch.~6 \& 5 respectively]{Fo10} for these
  notions. Generally determinantal point processes are more tractable, both from the viewpoint of theoretical
  analysis and of numerical computation, than Pfaffian point processes.
  
  We consider first a theoretical consequence. Consider an arc $J$ of the unit circle. Then the random variable
  $\mathcal N_J = \sum_{j=1}^N \mathbbm 1_{\theta_j \in J}$ counts the number of eigenvalues in $J$. Following
  \cite{CL95}, \cite{So00} one has that for any determinantal point process with an Hermitian kernel, and such
  that ${\rm Var} \, \mathcal N_J \to \infty$ as $N \to \infty$, one has
  \begin{equation}\label{5.5a}
  {\mathcal N_J  - \langle \mathcal N_J  \rangle \over ({\rm Var} \, \mathcal N_J )^{1/2}} \mathop{\to}\limits^{\rm d}
  \mathcal N(0,1),
    \end{equation} 
where $  \mathcal N(0,1)$ denotes the standard normal distribution.  On the other hand, the generating function
$\mathcal E^{{\rm ME}_N}(J; 1 - e^{i \omega})$, defined as in the first equation of (\ref{2.0d}) and in (\ref{2.30b})
is precisely the characteristic function for $\mathcal N_J$, since ${\rm Pr} (\mathcal N_J = k) = 
 E^{{\rm ME}_N}(k;J)$. Hence (\ref{5.5a}) implies that for large $N$
  \begin{equation}\label{5.5b} 
  \mathcal E^{{\rm ME}_N}(J; 1 - e^{i \omega}) = \exp \Big ( i \omega  \langle \mathcal N_J  \rangle - {\omega^2 \over 2}
  {\rm Var} \, \mathcal N_J + {\rm O}(1) \Big ).
  \end{equation}
  Application of (\ref{5.5b}) to the RHS of (\ref{2.31}) gives
  that (\ref{5.5b}) holds for the Pfaffian point process COE${}_N$,
 \begin{equation}\label{5.5c} 
  \mathcal E^{{\rm COE}_N}((0,\phi); 1 - e^{i \omega}) = \exp \Big ( i N \omega (\phi /(2 \pi))   - {\omega^2 \over 2}
  \sigma_N^2 + {\rm O}(1) \Big ), \quad \sigma_N^2 = {2 \over \pi^2}  \log \phi N,
  \end{equation} 
  valid provided that $\phi N \to \infty$ as $N \to \infty$. In the case of $\phi$ fixed, this result (extended to the general
  circular $\beta$ ensemble), is equivalent to a Gaussian fluctuation theorem of Killip \cite{Ki08}. One remarks too that
  the recent work \cite{FL20} extends, conditional on a still to be proved conjecture from \cite{FF04}, the expansion
  (\ref{5.5c}) and its $\beta$ generalisation to make explicit the O$(1)$ term; see also  \cite[Eq.~(49)]{SLMS21a}.
  
  In the case of the GUE, arguments based on (\ref{5.5a}) were used in \cite{Gu05} to establish Gaussian fluctuation formulas
  for the displacement of the eigenvalues from their expected positions, both with reference to the bulk, and the edge.
  For example, for a GUE eigenvalue $x_k$ near the centre of the spectrum (e.g.~$x_{(N-1)/2}$ for $N$ odd), one has the limit law
  \begin{equation}\label{5.5D}  
  \lim_{N \to \infty} \Big ( {2 \beta N \over \log N} \Big )^{1/2} x_k \mathop{=}^{\rm d} \mathcal N(0,1), \quad \beta = 2.
   \end{equation} 
   In \cite{OR10} O'Rourke used the inter-relation (\ref{5.3}) in the Gaussian case to extend the validity of such limit
   laws from the GUE to the GOE (now with $\beta = 1$).
  
  The generating functions on the RHS of (\ref{2.31}) can be written as Fredholm determinants \cite[Cor.~5.2]{BFM17},
  and also in terms of particular $\sigma$PVI $\tau$-functions \cite[\S 5.3]{BFM17}. The Fredholm determinant form is well
  suited to the numerical computation of  $\{ E_{N,1}(k;(0,\phi);{\rm COE}_N) \}$ \cite{Bo10}, while power series can be
  generated from the Painlev\'e expressions \cite[Eq.~(8.81)]{Fo10}. They too scale naturally in the
  bulk limit, when $\phi$ is replaced by $2 \pi s/N$ and the limit $N \to \infty$ is taken; see \cite[Prop.~3.2 and related text]{FW24}
  for a summary. In the case of the limiting Fredholm determinant formulas
  \begin{equation}\label{2.31A} 
  \lim_{N \to \infty} \mathcal E^{O^\pm(2N+1)}((0,\pi s/N);z)  :=    \mathcal E^{O^\pm}((0,s);z) = \det  (\mathbb I - z K_\infty^{\mp,(0,s)}),
  \end{equation}
where    $\mathcal K_\infty^{\pm,(0,s)}$ is the integral operator on $(0,s)$ with kernel
 \begin{equation}\label{2.31B} 
 K_\infty^\pm(x,y) =  K_\infty(x,y)  \pm K_\infty(x,-y) ,  \quad  K_\infty(x,y) := {\sin \pi (x - y) \over \pi (x - y)}=  {\rm sinc}\, \pi (x-y).
  \end{equation}
  the resulting characterisation of  bulk scaled conditioned gap probabilities $\{ E_{\infty,1}(k;(0,s) \}\}$ is equivalent to that
   given in \cite[Eq.~(20.1.20)]{Me04}. This characterisation is most conveniently written as the 
   bulk scaling limit of (\ref{2.31}) itself, which reads
   (\cite[Eqns.~(8.152)]{Fo10}) 
   \begin{equation}\label{5.5d}  
   \mathcal  E^{{\rm COE}_\infty} ((0,s); \xi)  = {(1 - \xi)  \mathcal E^{O^+}((0,s/2);\xi(2 - \xi)) +  \mathcal E^{O^-}((0,s/2);\xi(2-\xi)) \over 2 - \xi}.
  \end{equation}   
  
  Recent application of (\ref{5.5d}) has been made to the computation of the $N \to \infty$ limit of the power
  spectrum statistic; see the recent series of works \cite{ROK17,ROK20,RK23} for context and related 
  integrability results. For eigenangles $0 < \theta_1 < \cdots < \theta_N < 2 \pi$,
the power spectrum statistic is defined by the Fourier sum
 \begin{equation}\label{1.1z}
 S_N(\omega) = {1 \over N \Delta_N^2} \sum_{1 \le l,m\le N} \langle
 \delta \theta_l \delta \theta_m \rangle e^{i \omega ( l - m)}, \qquad \omega \in \mathbb R.
 \end{equation}
 The quantity $\delta \theta_l = \theta_l - \langle \theta_l  \rangle$ is the displacement from the mean $ \langle \theta_l  \rangle$
 and so $[\langle
 \delta \theta_l \delta \theta_m \rangle]_{l,m=1}^N$ is the covariance matrix of the level displacements. Further $\Delta_N := 2 \pi/N$ 
 is the mean spacing between eigenvalues. We have from \cite[Eq.~(3.7)]{FW24}
  \begin{equation}\label{3.2c}
  S_{\infty,1}(\omega) := \lim_{N \to \infty} S_{N}^{\rm COE}(\omega) = {1 \over  2 \sin^2 {\omega \over 2}} {\rm Re}
  \int_0^\infty  \mathcal E^{{\rm COE}_\infty} ((0,s); 1 - e^{i \omega}) \, ds.
   \end{equation} 
   As well as being tractable for the purposes of providing a tabulation of $S_{\infty,1}(\omega)$ (due to this being even in $\omega$ and periodic of period
   $2 \pi$, the range $0 < \omega \le \pi$ suffices), by use of an analogue of the asymptotic expansion (\ref{5.5c}) applying to
   $E^{{\rm COE}_\infty} ((0,s); 1 - e^{i \omega})$ in the limit $s \to \infty$, it is possible to show from (\ref{3.2c})
   that for $\omega \to 0^+$, 
  \begin{equation}\label{3.2d} 
   S_{\infty,1}(\omega) \sim {1 \over \beta \pi} {1 \over | \omega|} + {1 \over \beta^2 \pi^3} \log | \omega | + {\rm O}(\omega), \quad \beta = 1;
    \end{equation}   
   see \cite[Eq.~(3.11)]{FW24}.

  In the recent work \cite{RTK23} $S_\infty(\omega)$ has been shown to relate to the covariances of two nearest neighbour spacings between
bulk scaled eigenvalues. Thus define $x_l = (N / 2 \pi) \theta_l$, $(l=1,2,\dots,)$, so that the average of $x_{l+1} - x_{l}$ is unity.
 Define $s_j(p) = x_{j+1+p} - x_j$ as the spacing between $x_j$ and its $(p+1)$-th neighbour to the right $x_{j+1+p}$. The  covariance
 of interest is
$$
 {\rm cov}_\infty (s_j(0), s_{j+k+1}(0)) =  {1 \over 2} \Big ({\rm Var}_\infty(s_j(k+1)) - 2 {\rm Var}_\infty(s_j(k)) + {\rm Var}_\infty(s_j(k-1)) \Big ),
 $$
 where the equality follows from a result of  \cite{BFPW81}.
  Due to translation invariance, this is independent of $j$. The result of \cite{RTK23} gives that the power spectrum associated
  with these covariances relates to  $S_\infty(\omega)$ according to the general formula
  \begin{equation}\label{1.1x}
  \sum_{k=-\infty}^\infty {\rm cov}_\infty(s_j(0), s_{j+k}(0)) e^{i \omega k} = 4 \sin^2(\omega/2) S_\infty(\omega).
 \end{equation}  
 A consequence of the evaluation formula (\ref{3.2c}) in the case of the bulk scaled COE, is the asymptotic expansion \cite[Eq.~(11)]{RTK23}, \cite[Eq.~(3.12)]{FW24}
  \begin{equation}\label{3.9}   
  {\rm cov}_\infty  (s_j(0), s_{j+k}(0)) \mathop{\sim}\limits_{k \to \infty} - {1 \over \beta  \pi^2 k^2} - {6 \over \beta^2  \pi^4 k^4}
  \Big ( \log (2 \pi k) + \gamma - {11 \over 6} \Big ) + \cdots, \quad \beta = 1,
  \end{equation}
  where $\gamma$ denotes Euler's constant.

   \subsection{Decimations relating orthogonal symmetry and symplectic symmetry ensembles}
   Soon after the discovery of the inter-relation (\ref{2.0c}) by Dyson, Mehta and Dyson \cite[Eq.~(10)]{MD63}
   deduced, by making use of the method of integration over alternative variables \cite{Me04}, the inter-relation
 \begin{equation}\label{8.1}    
 {\rm alt} \, {\rm COE}_{2N}    \mathop{=}\limits^{\rm d}  {\rm CSE}_N.
   \end{equation}
   Here CSE stands for the circular symplectic ensemble (see e.g.~\cite[\S 2.2.2]{Fo10}), which has 
   eigenvalue PDF given by (\ref{1.1}) with $\beta = 4$.
   
 We recall from the paragraph including (\ref{3.0}) that we use the notation
  ${\rm SE}_N(w_4)$ to denote the $\beta = 4$ case of (\ref{3.0}).
   In the Hermitian case, a natural question to ask is for the specification of pairs of weights $(f,g)$ such that,
   for example even$\,({\rm OE}_{2N+1}(f)) \mathop{=}^{\rm d} {\rm SE}_N(g)$. In fact such a question reduces back to consideration
   of the previously considered decimation identities of superimposed ensembles, (\ref{5.3}) and (\ref{5.3x}).
   
   \begin{prop}\label{P3.2}
   (\cite{FR01}) The inter-relation identities
    \begin{equation}\label{L1} 
   {\rm even} \, ({\rm OE}_{2N}(f) \cup {\rm OE}_{2N+1}(f) )  \mathop{=}^{\rm d}  {\rm UE}_{2N}(g), \quad
    {\rm even} \, ({\rm OE}_{2N}(f) \cup {\rm OE}_{2N}(f) )  \mathop{=}^{\rm d}  {\rm UE}_{2N}(g)
   \end{equation}   
    are equivalent to the inter-relation identities
    \begin{equation}\label{L2} 
   {\rm even} \, ({\rm OE}_{2N+1}(f) )  \mathop{=}^{\rm d}  {\rm SE}_N((g/f)^2), \quad
    {\rm even} \, ( {\rm OE}_{2N}(f) )  \mathop{=}^{\rm d}  {\rm SE}_N((g/f)^2).
   \end{equation}  
   Similarly, the inter-relation identities
    \begin{equation}\label{L3} 
   {\rm odd} \, ({\rm OE}_{2N}(f) \cup {\rm OE}_{2N-1}(f) )  \mathop{=}^{\rm d}  {\rm UE}_{2N}(g), \quad
    {\rm odd} \, ({\rm OE}_{2N}(f) \cup {\rm OE}_{2N}(f) )  \mathop{=}^{\rm d}  {\rm UE}_{2N}(g)
   \end{equation}   
    are equivalent to the inter-relation identities
    \begin{equation}\label{L4} 
   {\rm odd} \, ({\rm OE}_{2N-1}(f) )  \mathop{=}^{\rm d}  {\rm SE}_N((g/f)^2), \quad
    {\rm odd} \, ( {\rm OE}_{2N}(f) )  \mathop{=}^{\rm d}  {\rm SE}_N((g/f)^2).
   \end{equation}  
   \end{prop} 
   
The analogue of (\ref{2.0d}) following from (\ref{8.1}) is \cite{MD63}   
   \begin{align}\label{8.140a}
\mathcal E_{N,4}^{\rm CSE}((-\theta, \theta);\xi) & 
  = {1 \over 2} \Big (
\mathcal E_{2N}^{+}((0,\theta);\xi) + \mathcal E_{2N}^{-}((0,\theta);\xi) \Big ) \nonumber \\
& = {1 \over 2} \Big (
\mathcal E^{O^+(2N+1)}((0,\theta);\xi) + \mathcal  E^{O^-(2N+1)}((0,\theta);\xi) \Big ),
\end{align}
where the second line follows from  (\ref{2.30}) and  (\ref{2.30a}); cf.~(\ref{2.31}).
Taking the bulk scaling limit gives 
 (\cite[Eqns.~(8.152)]{Fo10}) 
   \begin{equation}\label{5.5d+}  
   \mathcal  E^{{\rm CSE}_\infty} ((0,s); \xi)  =  {1 \over 2} \Big ( \mathcal E^{O^+}((0,s/2);\xi ) +  \mathcal E^{O^-}((0,s/2);\xi ) \Big );
  \end{equation}  
cf.~(\ref{5.5d}). The applications of    (\ref{5.5d}) discussed in \S \ref{S3.2} above all have counterparts for
(\ref{5.5d+}). For the details we refer to the previously cited  references.
  
  Matrix realisations of examples of the decimation inter-relations of Proposition \ref{P3.2} have been 
  given in \cite{FR02b}. In particular, both examples of realisations of particular
  interlaced ensembles and decimation identities given in \S \ref{S2.3} can be modified to apply
  to Proposition \ref{P3.2} \cite{FR02b}. Consider for definiteness the Gaussian case of the first identity in (\ref{L2}),
  which reads
  \begin{equation}\label{L2a} 
   {\rm even} \, ({\rm OE}_{2N+1}(e^{-x^2}) )  \mathop{=}^{\rm d}  {\rm SE}_N(e^{-x^2}) .
  \end{equation}  
  One begins with a complex Hermitian matrix, $A$ say, realisation of $ {\rm SE}_N(e^{-x^2})$ which up to a scale is
  the same as for the GSE; see \cite[Def.~1.3.2]{Fo10}. The matrix $A$ is of size $2N \times 2N$ and based on $2 \times 2$
  complex matrix representation of a quaternion; see e.g.~\cite[Eq.~(1.24)]{Fo10}, and furthermore
  is doubly degenerate. From this, the   $4N \times 4N$, four times degenerate, real symmetric matrix $\tilde{A}$ is formed by
 block  replacing each element according to (\ref{Q}). Next, a $(4N+1) \times (4N+1)$ bordered matrix of the structure
 (\ref{7.3y}) is formed with $\mathbf x$ equal to $\sqrt{b}$ times a $4N \times 1$ standard complex Gaussian vector, and
 the bottom right entry $a$ equal to $\sqrt{b}$ times a standard real Gaussian. In this setting, analogous to the
 result of Proposition \ref{P2.3}, the eigenvalues of $M$ are the distinct eigenvalues of $A$,
 $\{y_j\}_{j=1}^N $ say, as well as an additional $N+1$ eigenvalues $ \{x_j\}_{j=1}^{N+1}$, which must
 satisfy the interlacing (\ref{7.3y1}). Due to this interlacing, and by construction, one has
 \begin{equation}\label{7.3y3d} 
 {\rm even} \, (M ) \mathop{=}\limits^{\rm d}  {\rm SE}_N(e^{-x^2})
  \end{equation}
  (cf.~(\ref{7.3y3d})), and furthermore the joint eigenvalue PDF of the independent eigenvalues of
  $M$ is given by (\ref{7.3y2}) multiplied by the further factor
  $\prod_{j=1}^N \prod_{k=1}^{N+1} | x_j - y_k|$. Choosing $b = {1 \over 2}$ this is precisely the
  eigenvalue PDF of GOE${}_{2N+1}$, and so (\ref{7.3y3d}) provides a matrix theoretic realisation of
  (\ref{L2a}).

Another point of interest is that the Dixon-Anderson integral (\ref{DA}) contains as special
cases  both the identities in (\ref{L2}) with Jacobi weights. Moreover, from this starting point,
the Laguerre and Gaussian weight cases can be deduced by a limiting procedure; see
\cite[\S 4.2.3]{Fo10} for details.

\subsection{Decimations from a generalised Dixon-Anderson integral}
The  Dixon-Anderson integral (\ref{DA}) relates to integration over the interlaced region
(\ref{DA1}). A generalisation of this region containing a parameter $r$ is $A_r$ specified by
\begin{equation}\label{Aa}
a_j > \lambda_{r(j-1)+1} > \lambda_{r(j-1)+2} > \cdots >
\lambda_{rj} > a_{j+1} \qquad (j=1,\dots,N-1).
\end{equation}
Thus between every pair of variables $a_j, a_{j+1}$ there are $r$ variables $\lambda_i$.
In terms of this region, a generalisation of the Dixon-Anderson integral (\ref{DA})  has
been obtained in \cite{Fo07}.

\begin{prop}
For a known proportionality, one has
\begin{eqnarray}\label{LR}
\int_{A_r} d \lambda_1 \cdots d \lambda_{r(N-1)} \,
\prod_{1 \le j < k \le r(N-1)} (\lambda_j - \lambda_k)^{2/(r+1)}
\prod_{j=1}^{r(N-1)} \prod_{p=1}^N | \lambda_j - a_p|^{s_p - 1} \nonumber \\
\qquad \propto \prod_{1 \le j < k \le N} (a_j - a_k)^{r(s_j + s_k - 2/(r+1))}.
\end{eqnarray}
\end{prop}

Denote by ME${}_N(w_\beta)$ the Hermitian ensemble corresponding to (\ref{3.0}), as is
consistent with the notation used below (\ref{3.1a}), and further assume the ordering of
eigenvalues $x_1 > x_2 > \cdots > x_N$. Let D${}_r$ denote the operation of 
observing only those eigenvalues labelled by a multiple of
$r$.
As a corollary of (\ref{LR}) one has \cite{Fo07}, \cite[Prop.~4.4.2]{Fo10}
\begin{eqnarray}
&& {\rm D}_{r+1} ({\rm ME}_{2/(r+1), (r+1)N + r} ((1+x)^a (1 - x)^b)) =
{\rm ME}_{2(r+1) , N} ( (1+x)^{(r+1) a + 2r} (1 - x)^{(r+1)b + 2r} ) 
\nonumber  \\
&& {\rm D}_{r+1} ({\rm ME}_{2/(r+1), (r+1)N } ((1 - x)^b)) =
{\rm ME}_{2(r+1), N} ((1 - x)^{(r+1)b + 2r} ),
\label{3.u.2}
\end{eqnarray}
where the weights herein are supported on $-1<x<1$. In the case $r=1$ these
inter-relations correspond to the Jacobi case of the 
 identities in (\ref{L2}). By taking certain limits, Laguerre and Gaussian
 versions can also be obtained. 
 
 There is also a circular analogue of the Dixon-Anderson integral 
 \cite{FR06}, and this in turn permits a generalisation analogous to
 (\ref{LR}). In terms of the notation ${\rm CE}_N(\beta)$ for the circular
 $\beta$ ensemble (\ref{1.1}), a notable consequence of the latter is
 the generalisation of (\ref{8.1})
 \begin{equation}\label{3.111e}
{\rm alt}_{r+1} ( {\rm CE}_{(r+1) N}(2/(r+1)) ) = {\rm CE}_N(2(r+1)),
\end{equation}
where alt${}_{r+1}$ denotes the distribution of any sequence of $N$ eigen-angles,
observed from the original ensemble of $(r+1)N$ eigenvalues,
chosen with regular spacing of $r$ eigenvalues in between. 

Although not
the main focus of this review, one should remark that the inter-relations 
(\ref{3.u.2}) and (\ref{3.111e}) 
can also be interpreted as  examples of duality relations, in which the Dyson
parameter $\beta$ is mapped from $2/(r+1)$ to $2(r+1)$. In particular $r = 1$ 
corresponds to a duality between an ensemble with orthogonal symmetry, and
an ensemble with symplectic symmetry.

\section{Inter-relations involving the spectral form factor and discrete determinantal point processes}
\subsection{The spectral form factor}
Our attention to date has been focussed on random matrix inter-relations
involving superposition or decimations, which as we have seen admit a substantial theory.
Another topic in random matrix theory which gives rise to systematic inter-relations
is the computation of the spectral form factor (also known as structure function)
 \begin{equation}\label{Sa0}
S_N(k) = \Big \langle \Big |  \sum_{j=1}^N e^{i k x_j} \Big |^2  \Big \rangle -
 \Big |   \Big \langle  \sum_{j=1}^N e^{i k x_j}   \Big \rangle  \Big |^2.
\end{equation}
Note that this quantity can be considered as the variance of the (complex) linear statistic
$\sum_{j=1}^N e^{i k x_j}$. As such it relates to the two-point correlation $\rho_{(2)}(x,x')$ and
the one-point density $\rho_{(1)}(x)$ according to the double integral
\begin{equation}\label{Sa1}
S_N(k) = \int_{\Omega} dx \, e^{i k x}  \int_{\Omega} dy \, e^{-i k y} \Big ( \rho_{(2)}(x,y) 
+ \delta(x-y) \rho_{(1)}(y) - \rho_{(1)}(x) \rho_{(1)}(y) \Big ),
\end{equation}
where $\Omega$ denotes the support of the spectrum. The quantity $ \rho_{(2)}(x,y) -  \rho_{(1)}(x) \rho_{(1)}(y)$
is referred to as the truncated two-point correlation, or as the two-point cluster function and denoted
$\rho_{(2)}^T(x,y)$. In the
case of the circular ensemble, for $k$ integer rotational invariance allows (\ref{Sa1}) to be reduced from
a double integral to a single integral
$$
S_N(k) =  2 \pi \int_{-\pi}^\pi e^{i k \theta}  \rho_{(2)}^T(\theta,0) \, d \theta - N.
$$
It is the bulk scaled limit  of this single integral  which
was computed explicitly in the pioneering work of Dyson \cite{Dy62a}.

Inter-relations between ensembles show themselves in the computation of $S_N(k)$ for the
GUE \cite{BH97}, the LUE \cite{Fo21b}, and most recently for  the elliptic Ginibre ensemble
\cite{SK24}. We will consider the GUE and LUE cases first, which exhibit inter-relations with the
LUE and JUE respectively. With regards to the latter, we will use the notation ${\rm JUE}_{N,(a,b)}$
for the ensemble ${\rm UE}_N(x^a(1-x)^b \mathbbm 1_{0 < x < 1})$.

\begin{prop}\label{P4.1}
One has the evaluation formulas
 \begin{align}
S_N^{{\rm GUE}_N}(k) & =  \int_0^k t \rho_{(1)}^{{\rm LUE}_{N,0}}(t^2/2) \, dt \nonumber \\
S_N^{{\rm LUE}_{N,\alpha}}(k) & =   \int_{1/(1+k^2)}^1  \rho_{(1)}^{{\rm JUE}_{N,(\alpha,0)}}(t) \, dt. 
\end{align}
\end{prop}

\begin{proof} (partial working  only)
The derivation of the result for $S_N^{{\rm LUE}_{N,\alpha}}(k)$ follows ideas used in the derivation of
the evaluation formula for $S_N^{{\rm GUE}_N}(k)$ due to Okuyama \cite{Ok19}, or more precisely a
reworking of those ideas given in \cite{Fo21a}. We will therefore focus attention on this latter
calculation. Due to the eigenvalues of the GUE and LUE being
determinantal point processes, with kernels that furthermore can be simplified according to the
Christoffel-Darboux summation (see \cite[Prop.~5.1.3]{Fo10}), the task of evaluating
(\ref{Sa1}) reduces to that of computing
$
I_N(k):= \int_{\Omega} dx \,  e^{i k x}  \int_{\Omega} dy \,  e^{-i k y}  ( 
 K_N(x,y)
   )^2$, where
 $$ 
          K_N(x,y):=(w_2(x) w_2(y))^{1/2}  {p_N(x) p_{N-1}(y) - p_{N-1}(x) p_N(y) \over \langle p_{N-1}, p_{N-1} \rangle (x - y)}.
$$
Here $\{p_j(x) \}$ are the (monic) orthogonal polynomials corresponding the (Gaussian or
Laguerre) weight $w_2(x)$ with inner product $\langle f,g \rangle := \int_{\Omega} w_2(x) f(x) g(x) \, dx$.

An important detail is to generalise the definition of $I_N(k)$ to a function of two variables
$I_N(t_1,t_2):= \int_{\Omega} dx \,  e^{i t_1 x}  \int_{\Omega} dy \,  e^{i t_2 y}  ( 
 K_N(x,y))^2$.
With this done one notes 
\begin{multline}\label{Sa1.1}
(t_1 - t_2) \Big ( {\partial \over \partial  t_1} + {\partial \over \partial  t_2} \Big )
 I_N(t_1,t_2 ) =  
 \int_{\Omega} dx \,  \int_{\Omega} dy \, (x-y) \Big (  \Big ( {\partial \over \partial x} + {\partial \over \partial y} \Big )
e^{it_1x+it_2y}  \Big )  ( 
 K_N(x,y)
   )^2.
 \end{multline} 
   Although this is true in general, its use is particular to the Gaussian case due to the identity \cite{TW94c} \cite[\S 5.4.2]{Fo10}
 \begin{equation}\label{Sa1.2}  
    \Big ( {\partial \over \partial x} + {\partial \over \partial y} \Big )
 K_N^{( G)} (x,y) = - {1 \over  \langle p_{N-1}, p_{N-1} \rangle^{(G)}} \Big ( \psi_N^{(G)}(x)  \psi_{N-1}^{(G)}(y) + \psi_{N-1}^{(G)}(x)  \psi_{N}^{(G)}(y) \Big ),
  \end{equation} 
 where $ \psi_{n}^{(G)}(x) := (w_2(x))^{1/2} p_n^{(G)}(x)$; see \cite[\S 3.2]{Fo21b} for the replacement working
 in the Laguerre case.  Integration by parts in (\ref{Sa1.1}) and  use of (\ref{Sa1.2}) then shows that the RHS of the former  is equal to
 \begin{equation}\label{Sa1.3}   
 {2 \over (  \langle p_{N-1}, p_{N-1} \rangle^{(G)})^2} \int_{\mathbb R^2} dx dy \, e^{it_1x+ it_2y} \Big (
 ( \psi_N^{(G)}(x)  \psi_{N-1}^{(G)}(y) )^2 -  ( \psi_N^{(G)}(y)  \psi_{N-1}^{(G)}(x) )^2 \Big ).
  \end{equation} 
  The integrals herein factorise into one-dimensional integrals, each of which can be evaluated in
  terms of the product of a Gaussian, a power function, and a Laguerre polynomial in squared
  variables; see \cite[Prop.~13]{Fo10} and references therein. This then allows the RHS of (\ref{Sa1.1}) to
  be identified as proportional to  $(t_`^2 - t_2^2)K^{{\rm LUE}_{N,0}}(t^2_1/2,t^2_2/2)$. After cancelling a
  factor of $(t_1 - t_2)$ from both sides, the limit $t_2 \to -t_1 = -t$ can now be taken,
  with the LHS reducing to ${d \over dt}I_N(t)$. On the RHS,
  noting that $K^{{\rm LUE}_{N,0}}(t^2_1/2,t^2_2/2)$ is then equal to $ \rho_{(1)}^{{\rm LUE}_{N,0}}(t^2/2)$,
  the stated evaluation
  formula for $S_N^{{\rm GUE}_N}(k)$ now follows upon a single integration.
\end{proof}

Next we turn to consider the analogue of (\ref{Sa0}) for the Ginibre ensemble of complex Gaussian matrices.
Up to proportionality, the eigenvalue PDF for the latter is (see e.g.~\cite{BF22})
$$
\prod_{l=1}^N e^{- | z_l |^2} \prod_{1 \le j < k \le N} | z_k - z_j|^2, \quad z_l \in \mathbb C.
$$
Upon the scaling $z_l \mapsto \sqrt{N} z_l$ the leading eigenvalue support is the unit disk; let this scaled
version of the Ginibre ensemble be denoted GinUE${}^*$. The so-called (connected) dissipative spectral form 
factor associated with this ensemble is specified by \cite{LPC21}
 \begin{equation}\label{Sa1.4}
 S_N^{\rm GinUE^*}(|T|) = \Big \langle \Big |  \sum_{j=1}^N e^{i \bar{T} z_j + i T \bar{z}_j}  \Big |^2  \Big \rangle -
 \Big |   \Big \langle  \sum_{j=1}^N e^{i \bar{T} z_j}   \Big \rangle  \Big |^2, \quad T \in \mathbb C.
\end{equation}
Its evaluation as a double sum involving Laguerre polynomials is known from \cite{LPC21}, \cite{GSV23}.
Very recently \cite{SK24}, the latter computation has been generalised to the elliptic Ginibre ensemble of
asymmetric complex Gaussian matrices (denoted eGinUE; see \cite[\S 2.3]{BF22}). By varying the parameter therein, the original
GinUE${}^*$ is obtained in one limit, and the GUE in another, the latter with the scaling $x_j \mapsto \sqrt{N \over 2} x_j$.
Denote the GUE so scaled by GUE${}^*$. With this as the ensemble average of (\ref{Sa1}) we
write $ S_N^{\rm GUE^*}(k)$. Comparing the two limiting expressions revealed an inter-relation between
$ S_N^{\rm GinUE^*}(|T|) $ and $ S_N^{\rm GUE^*}(k)$ (the latter, up to scaling, admitting the evaluation
formula of Proposition \ref{P4.1}).

\begin{prop} (\cite[Eq.~(23)]{SK24})
We have
 \begin{equation}\label{Sa1.5}
 S_N^{\rm GinUE^*}(|T|) = e^{-|T|^2/4}  S_N^{\rm GUE^*}(|T|/2).
 \end{equation}
 \end{prop}
 
 In the case of the elliptic Ginibre ensemble, the average in (\ref{Sa1.4}) is no longer
 dependent on just $|T|$, but rather $t,s$ with $T=t+is$, so one has $S_N^{\rm eGinUE^*}(t,s)$.
 Nonetheless it was shown in \cite[Eq.~(24)]{SK24} that this latter quantity is related to
 $S_N^{\rm GUE^*}(k)$ by a suitable generalisation of (\ref{Sa1.5}).

 \subsection{Discrete determinantal point processes}
 Generally a determinantal point process on the real line is characterised by correlation kernel
 $K^{\rm c}(x,y)$ (we use the superscript ``c'' to indicate that the points take on continuous values)
 such the general $k$-point correlation function is given by
 \begin{equation}\label{D1} 
 \rho_{(k)}(x_1,\dots,x_k) = \det [ K^{\rm c}(x_{j_1}, x_{j_2})]_{j_1,j_2=1,\dots,k}.
 \end{equation}
 Denote by $E^{\rm c}(k;J)$ the  probability that for this point process the interval $J$ contains exactly $k$
 points, and denote by $\mathcal E^{\rm c}(J;\xi)$ the  generating function
 $\mathcal E^{\rm c}(J;\xi) := \sum_{k=0}^\infty (1 - \xi)^k E^{\rm c}(k;J)$. Then it is a standard fact
 (see e.g.~\cite[Eq.~(9.15)]{Fo10}) that
 \begin{equation}\label{D3}  
 \mathcal E^{\rm c}(J;\xi)  = \det ( \mathbb I - \xi \mathcal K_J^{\rm c}),
 \end{equation} 
where $ \mathcal K_J^{\rm c}$ is the integral operator on the interval $J$ with kernel
$ K^{\rm c}(x,y)$. The meaning of det in (\ref{D3}) is  the Fredholm determinant of the integral
operator, which is equal to the product $\prod_{l}(1 - \xi \lambda_l)$, where $\lambda_l$ are
the eigenvalues of $ \mathcal K_J^{\rm c}$.

Instead of taking continuous values, in the above setting suppose the points must take on discrete values
from $\{y_l\}$, and let the correlation kernel then be denoted  $K^{\rm d}(x,y)$. Then, in addition to
the analogue of (\ref{D3}) one has for the generating function
\begin{equation}\label{D4} 
 \mathcal E^{\rm d}(J;\xi)  = \det [ \mathbb I_{|J|} - \xi K_J^{\rm d}(y_{j_1},y_{j_2})]_{\{y_{j_1}\}, \{ y_{j_2} \} \in J},
 \end{equation} 
where it is now understood that $J$ is discrete set contained in the support, and det refers to the usual determinant.
 The interplay between (\ref{D3}) and
(\ref{D4}) underpins a class of duality relations linking a random matrix gap probability to a gap probability for
a determinantal point process on a discrete set \cite{BO17,KS22}. Here we will present just one example of this
type.

\begin{prop}
(\cite[Th.~3.7]{BO17} for the case $\xi = 1$, \cite[\S 3.1.2]{KS22} in general)
Let LUE${}_{N,\alpha}$ denote the Laguerre unitary ensemble as previously defined. . Denote by
dLUE${}_{N,\alpha}$ the discrete Laguerre unitary ensemble specified as the discrete determinantal point process
defined on the positive integer lattice with correlation kernel
$$
K^{\rm d}(n_1,n_2) = \int_s^\infty \psi_{n_1}^{\rm L}(x)  \psi_{n_2}^{\rm L}(x) \, dx, \quad  \psi_{n}^{\rm L}(x) :=
\Big ( {n! \over \Gamma(n+\alpha + 1)} \Big )^{1/2} x^{\alpha/2} e^{-x/2} L_n^{(\alpha)}(x),
$$
where $s>0$ is a parameter.
We have
\begin{equation}\label{D5} 
\mathcal E^{{\rm c}, \, {\rm LUE}_{N,\alpha}}((s,\infty);\xi) = \mathcal E^{{\rm d}, \, {\rm dLUE}_{N,\alpha}}(\{0,1,\dots,N-1\};\xi).
 \end{equation} 
 \end{prop}
 
 \begin{proof}
Use of the formula \cite[Prop.~8.1.2]{Fo10} gives
\begin{multline}
\mathcal E^{{\rm c}, \, {\rm LUE}_{N,\alpha}}((s,\infty);\xi)  \propto \prod_{l=1}^N \int_0^\infty dx_1 \, (1 - \xi \mathbbm 1_{x_1 \in (s,\infty)}) \cdots
 \int_0^\infty dx_N \, (1 - \xi \mathbbm 1_{x_N \in (s,\infty)})  \\
  \times \prod_{l=1}^N x_l^\alpha e^{-x_l} \prod_{1 \le j < k \le N} (x_k - x_j)^2.
\end{multline}
By use of the Vandermonde determinant identity we have that the second line in this expression can be written, up to
proportionality, as $( \det [ \psi_{j-1}^{\rm L}(x_l) ]_{j,l=1,\dots,N})^2$ (see \cite[Proof of Prop.~5.1.1 for the required working]{Fo10}).
Use now of Andr\'eief's integration formula \cite{Fo19} and  the orthogonality $\int_0^\infty \psi_{j_1-1}^{\rm L}(x)  \psi_{j_2-1}^{\rm L}(x) \, dx =
\delta_{j_1,j_2}$ gives
$$
\mathcal E^{{\rm c}, \, {\rm LUE}_{N,\alpha}}((0,s);\xi)  = \det \Big [ \mathbbm 1_{N} - \xi
\int_s^\infty \psi_{j_1-1}(x) \psi_{j_2 - 1}(x) \, dx \Big ]_{j_1,j_2=1,\dots,N}.
$$
Comparison with (\ref{D4}) implies the stated result.
  \end{proof}
  
  There are other examples of identities between random matrix averages and gap probabilities for
   discrete determinantal point processes. One, which in a random matrix context was isolated by
   Borodin and Okounkov \cite{BO00} in the year 2000, and earlier still by Geronimo and Case \cite{GC79}
   in the context of Toeplitz determinant theory, identifies a discrete kernel (in the form of a double contour integral) of
   a Wiener-Hopf (discrete Fredholm) operator $\mathcal K_{\{N+1,N+2,\dots \}}^{\rm d}$ such that
 \begin{equation}\label{D6}
 \Big \langle e^{ \sum_{j=1}^N c(\theta_j)} \Big \rangle_{{\rm CUE}_N} = e^{N c_0 + \sum_{p=1}^\infty p c_{-p} c_p}
 \det ( \mathbb I -   \mathcal K_{\{N+1,N+2,\dots \}}^{\rm d}).
 \end{equation}  
 A concrete example, obtained in the original paper \cite[Example 1]{BO00} reads
  \begin{equation}\label{D7}
 \Big \langle e^{ s \sum_{j=1}^N \cos \theta_j } \Big \rangle_{{\rm CUE}_N} =  e^{s^2/4}  \det ( \mathbb I -   \mathcal K_{\{N+1,N+2,\dots \}}^{\rm dB}),
 \end{equation}  
 where the superscript ``dB'' refers to the discrete Bessel kernel \cite{BOO00}, defined for $s > 0$ by
   \begin{equation}\label{D8}
 K^{\rm dB}(n_1,n_2) =  s {J_{n_1}(s) J_{n_2+1}(s) - J_{n_1+1}(s) J_{n_2}(s) \over 2 (n_1 - n_2)}.
 \end{equation}  
 
 Identities relating to (\ref{D7}) can be further developed. Specifically, the class of random matrix dualities not part of this review,
 involving averages over characteristic polynomials (recall the opening paragraph of the Introduction, and the example
 (\ref{1.0})), can be used to show \cite{Fo93c}
   \begin{equation}\label{D9} 
  \Big \langle e^{ s \sum_{j=1}^N \cos \theta_j } \Big \rangle_{{\rm CUE}_N} =  e^{s^2/4}   \det ( \mathbb I - \mathcal K_{(0,s^2)}^{\rm B} ), \quad \alpha \in \mathbb Z_{\ge 0},
  \end{equation}  
 where the superscript ``B'' refers to the (continuous) Bessel kernel \cite{Fo93a}
   \begin{equation}\label{D10}
 K^{\rm B}(x,y) = {J_\alpha(\sqrt{x}) \sqrt{y} J_\alpha'(\sqrt{y}) -  \sqrt{x}) J_\alpha'(\sqrt{x})J_\alpha(\sqrt{y}) \over 2 (x - y) }.
 \end{equation}   
 An identity between continuous and discrete Fredholm determinants implied by equating (\ref{D9}) with (\ref{D7}) in the case $N = \alpha$.
 In fact this identity holds more generally with the inclusion of a $\xi$ parameter,
  \begin{equation}\label{D11} 
   \det ( \mathbb I -  \xi  \mathcal K_{\{\alpha +1,\alpha +2,\dots \}}^{\rm dB}) =   \det ( \mathbb I - \xi \mathcal K_{(0,s^2)}^{\rm B} ), \quad \alpha \in \mathbb Z_{\ge 0},
 \end{equation}   
 as established by Moriya, Nagao and Sasamoto \cite{MNS19}.

  \subsection*{Acknowledgements}
	This work  was supported
	by the Australian Research Council 
	 Discovery Project grant DP210102887.

\providecommand{\bysame}{\leavevmode\hbox to3em{\hrulefill}\thinspace}
\providecommand{\MR}{\relax\ifhmode\unskip\space\fi MR }
\providecommand{\MRhref}[2]{%
  \href{http://www.ams.org/mathscinet-getitem?mr=#1}{#2}
}
\providecommand{\href}[2]{#2}

  \end{document}